\input ./style/arxiv-vmsta.cfg
\documentclass[numbers,compress,v1.0.1]{vmsta}

\usepackage{dsfont}

\volume{2}
\issue{3}
\pubyear{2015}
\firstpage{233}
\lastpage{249}
\doi{10.15559/15-VMSTA36CNF}

\newtheorem{thm}{Theorem}[section]
\newtheorem{lemma}{Lemma}[section]

\theoremstyle{definition}
\newtheorem{remark}{Remark}[section]
\newtheorem{defin}{Definition}[section]

\startlocaldefs

\allowdisplaybreaks
\endlocaldefs

\begin{document}
\begin{frontmatter}

\title{Pricing the European call option in the model with stochastic
volatility driven by
Ornstein--Uhlenbeck process. Exact formulas}

\author{\inits{S.}\fnm{Sergii}\snm{Kuchuk-Iatsenko}\corref{cor1}}\email
{kuchuk.iatsenko@gmail.com}
\cortext[cor1]{Corresponding author.}
\author{\inits{Yu.}\fnm{Yuliya}\snm{Mishura}}\email{myus@univ.kiev.ua}
\address{Taras Shevchenko National University of Kyiv\\
Volodymyrska str. 64, 01601
Kyiv, Ukraine}

\markboth{S.~Kuchuk-Iatsenko, Yu.~Mishura}{Evaluation of the price of
European call option}

\begin{abstract}
We consider the Black--Scholes model of financial market modified to
capture the stochastic nature of volatility observed at real financial
markets. For volatility driven by the Ornstein--Uhlenbeck process, we
establish the existence of equivalent martingale measure in the market
model. The option is priced with respect to the minimal martingale
measure for the case of uncorrelated processes of volatility and asset
price, and an analytic expression for the price of European call
option is derived. We use the inverse Fourier transform of a
characteristic function and the Gaussian property of the
Ornstein--Uhlenbeck process.
\end{abstract}



\received{29 July 2015}
\revised{13 September 2015}
\accepted{14 September 2015}
\publishedonline{25 September 2015}
\end{frontmatter}

\section{Introduction}\label{sec1}

One of the promising directions of enhancement of the classical
Black--Scholes model is construction and research of diffusion models
with volatility of risky asset governed by a stochastic process.
Empirical studies \cite{paper7,paper13} evidence in favor of
the fact that the classical model with constant volatility is unable to
capture important features of volatility observed in real financial
markets. This drawback of the Black--Scholes model has been widely
investigated and to some extent eliminated by the extension of the
theory in three directions: models with time-dependent deterministic
volatility, models with state-dependent volatility, and models with
stochastic volatility. The first and second of these categories may be
viewed as intermediate between the classical model and third category,
although equipping the market with certain constraints (the most
essential is the limiting time period under consideration) allow less
complex models to produce results of acceptable precision.

Despite recent popularity of the stochastic volatility modification of
the Black--Scholes theory, the range of models under consideration is
quite narrow. One of the first models of such a type is presented in
\cite{paper8}, where the authors assume the volatility of the price of
risky asset to be governed by the square root of the geometric Brownian
motion. An expression for the price of European call option is derived
under the following assumption: the volatility process is driven by a
Brownian motion independent of the Brownian motion governing the price
of risky asset. In \cite{paper9}, the authors choose the
Ornstein-Uhlenbeck (OU) process to drive the volatility. The OU process
is mean-reverting, and there is a strong evidence that the volatility
in real financial markets has such a feature \cite{paper1,book1}. Under this assumption, the authors of \cite{paper9} describe
the distribution of the price of risky asset and apply it to derive an
estimate of the price of European call option. As an alternative, there
is an option to choose the Cox--Ingersoll--Ross process to govern the
volatility process \cite{paper1,book1}. It is worth
mentioning that although all cited works contain some significant
results, they rely upon simplified models of real-world volatility
process (e.g., ignoring the mean-reversion property). We remark that
there is a vast amount of further investigations that consider more
sophisticated and thus more realistic models. An extensive overview of
these results is given in \cite{paper16}.

Nonnegativity is another desirable feature of the process modeling
volatility. One of possible choices is to use the exponential function
of the OU process (see \cite{paper11,paper10}, and references therein).

Questions of existence of equivalent (local) martingale measures are
investigated in different frameworks and different generality in \cite
{book1,paper2,book3,paper14}. Often, after
specifying the model, the authors state that a risk-neutral measure
exists and continue investigation in the risk-neutral world without
defining the measure.

A significant part of works (including aforementioned) use the Fourier
transform to derive an analytical representation of the price of
European call option. A great deal of information about developments in
application of the Fourier transform to option pricing problems can be
found in \cite{paper15}.

Our work investigates the market defined by a diffusion model with
stochastic volatility being an arbitrary function governed by the
Ornstein--Uhlenbeck process. Under general setting and quite mild
assumptions, we prove that the market satisfies two distinct
no-arbitrage properties for different classes of trading strategies.
For the special case of uncorrelated Wiener processes, we derive an
analytical expression for the price of European call option.

This paper is structured as follows: in Section~\ref{sec2}, we define a general
model. In Section~\ref{sec3}, we present definitions and preliminary results
necessary for further analysis. In Section~\ref{sec4}, we investigate matters of
existence of equivalent (local) martingale measures and arbitrage
properties of the general model. In Section~\ref{sec5}, we define a particular
case of the general model and raise the problem of pricing European
call option. Section~\ref{sec6} covers the derivation of an~analytical
expression for the option price.

\section{Diffusion model with stochastic volatility governed by
Ornstein--Uhlenbeck process}\label{sec2}

Let $\{\varOmega, \mathcal{F}, \mathbf{F}=\{\mathcal{F}_t^{(B,W)}, t\geq
0\}, \mathbb{P}\}$ be a complete probability space with filtration
generated by correlated Wiener processes $\{B_t$, $W_t$, $0 \leq t \leq
T\}$. We consider the model of the market where one risky asset is
traded, its price evolves according to the geometric Brownian motion $\{
S_t,\;0 \leq t \leq T\} $, and its volatility is driven by a stochastic
process. More precisely, the market is described by the pair of
stochastic differential equations

\begin{equation}
\label{Model0} dS_t = \mu S_tdt+\sigma(Y_t)S_tdB_t,
\end{equation}
%
\begin{equation}
\label{Model1} dY_t = -\alpha Y_tdt+kdW_t.
\end{equation}

Denote by $S_0=S $ and $Y_0=Y$ the deterministic initial values of the
processes specified by Eqs.~\eqref{Model0}--\eqref{Model1}.

In Sections~\ref{sec2}--\ref{sec4}, we impose the following assumptions:

\begin{itemize}

\item[(A1)] The Wiener processes $B$ and $W$ are correlated with
correlation coefficient $\rho\in[-1;1] $, that is, $dB_tdW_t=\rho dt$;

\item[(A2)] the volatility function $\sigma: \mathbb{R}\rightarrow
\mathbb{R}_+$ is measurable, bounded away from zero by a constant $c$,
that is,
\[
\sigma(x) \geq c > 0,\quad  x \in\mathds{R},
\]
and satisfies the conditions $\int_0^T\sigma^2(Y_t)dt<\infty$ a.s.;

\item[(A3)] the coefficients $\alpha$, $\mu$, and $k$ are positive.
\end{itemize}

For example, the conditions mentioned in assumption (A2) are satisfied
for a~measurable function $\sigma(x)$ such that $c \leq\sigma^2(x)
\leq C$ for $x\in\mathbb{R}$ and some constants $0<c<C$. Moreover,
given the square integrability of $\sigma(Y_s)$, the solution of the
differential equation \eqref{Model0} is given by
%
\begin{equation}
S_t=S_0\exp{ \Biggl(\mu t - \frac{1}{2}\int
^t_0\sigma^2(Y_s)ds+\int
^t_0\sigma(Y_s)dB_s
\Biggr)},
\end{equation}
which yields that $S_t$ is continuous. Hence, the product $\sigma
(Y_s)S_t$ is square integrable: $\int_0^T\sigma^2(Y_t)S^2_tdt<\infty$ a.s.

The unique solution of the Langevin equation \eqref{Model1} $Y_t$ is
the so-called Ornstein--Uhlen\-beck (OU) process. Its properties make
it a suitable tool for modeling the volatility in financial markets.
One of the most important features is the mean-rever\-sion property. The
OU process is Gaussian with the following characteristics:
\begin{eqnarray*}
E[Y_t]=Y_0\operatorname{e}^{-\alpha t}, \qquad \operatorname {Var}[Y_t]=\frac{k^2}{2\alpha}\bigl(1-\operatorname{e}^{-2\alpha t}\bigr).
\end{eqnarray*}

Moreover, the OU process is Markov and admits the explicit representation
\[
Y_t=Y_0 \operatorname{e}^{-\alpha t}+k\int
_{0}^t \operatorname {e}^{-\alpha(t-s)}dW_s.
\]

We can represent the process $W$ in the form
\[
W_t=\rho B_t+\sqrt{1-\rho^2}
Z_t,
\]
where $Z$ is a Wiener process independent of $B$. In what follows, we
will use this representation. Notice that $\mathcal{F}^{(B,W)}= \mathcal
{F}^{(B,Z)}$, where the filtration $\{\mathcal{F}^{(B,Z)}_t,0 \leq t
\leq T\}$ is generated by independent Wiener processes $B$ and $Z$.

\section{Definitions and preliminary results}\label{sec3}

Most of the information presented in this section can be found in more
detail in \cite{book2,book5} (and other references below).

We consider the market with one risky asset and one risk-free asset.
Evolutions of prices of both assets are given by a semimartingale
process $(\hat{S}_t)_{t=0}^T$ and deterministic process
$(B_t)_{t=0}^{T}=\operatorname{e}^{rt}$, respectively, where $r$ is a
constant risk-free rate of return. We introduce the discounted price
process $({S}_t)_{t=0}^T=\operatorname{e}^{-rt}\hat{S}_t$.

Agents acting in the market may buy or sell risky asset and make their
decisions concerning the structure of their portfolios basing upon the
information available at the moment of decision. This principle can be
formalized by the following definition.
\begin{defin}\label{TS}
A trading strategy is a predictable process $(\pi_t)_{t=0}^T$. The
value~$\pi_t$ of this process represents the amount of asset $\hat{S}$
in a portfolio at time $t$.
\end{defin}

Certain amount of preliminary concepts is necessary in order to
introduce the essential notion of admissible self-financing strategy.
Let a semimartingale~$S$ admit the decomposition $S=S_0+A+M$, where $A$
is a bounded-variation process, and $M$ is a local martingale.
According, for example, to \cite{book5}, p.~635, there is a
nondecreasing adapted (to the filtration $(\mathcal{F}_t)_{t \geq0}$)
process $C=(C_t)_{t \geq0}$, $C_0=0$, and adapted processes
$c=(c_t)_{t \geq0}$ and $\hat{c}=(\hat{c}_t)_{t \geq0}$ such that
\begin{equation*}
A_t=\int_0^t c_s
dC_s, \quad t>0,
\end{equation*}
and the quadratic variation equals
\begin{equation*}
[M,M]_t=\int_0^t
\hat{c}_s dC_s.
\end{equation*}

\begin{defin}\label{TS2}
Let $\pi$ be a predictable process. We shall say that:
\begin{itemize}
\item$\pi\in L_{\mathit{var}}(A)$ if for all $\omega\in\varOmega$, we have $\int_0^t \pi_s c_s dC_s<\infty$, $t>0$;
\item$\pi\in L_{\mathit{loc}}^q(M)$, $q \geq1$, if there exists a sequence of
stopping times $\tau_n$ approaching $\infty$ as $n \rightarrow\infty$
such that
\begin{equation*}
\mathbb{E} \Biggl[\int_0^{\tau_n} \pi^2
\hat{c}_s dC_s \Biggr]^{q/2}<\infty;
\end{equation*}
\item$\pi\in L^q(S)$ if there exists a representation $S=S_0+A+M$
such that $\pi\in L_{\mathit{var}}(A) \cap L_{\mathit{loc}}^q(M)$.
\end{itemize}
\end{defin}

\begin{defin}\label{TS3}
A trading strategy is called admissible (relative to the price process
$S$) if $\pi\in L^1(S)$.
\end{defin}

\begin{defin}\label{TS4}
An admissible strategy is said to be self-financing (relative to the
price process $S$) or, equivalently, $\pi\in \mathit{SF}(S)$ if its value
$S_t^{\pi}=\pi_t S_t$ has a~representation $S_t^{\pi}=S_0^{\pi}+\int_0^t \pi_s dS_s$.
\end{defin}

Further, we define two particular classes of trading strategies along
with the corresponding classes of $\mathcal{F}_T$-measurable pay-off
functions $\psi=\psi(\omega)$ that can be majorized by returns of
strategies belonging to each class.

\begin{defin}\label{TSPa}
For each $a\geq0$, define
\begin{equation*}
\varPi_a(S)=\bigl\{\pi\in \mathit{SF}(S): S_t^{\pi}
\geq-a, \ t \in[0,T]\bigr\}
\end{equation*}
and
\begin{equation*}
\varPsi_{+}(S)=\Biggl\{\psi\in L_{\infty}(\varOmega,
\mathcal{F}_T,\mathbb{P}) : \psi\leq\int_0^T(
\pi_s,dS_s) \ \text{for some} \ \pi\in
\varPi_{+}(S)\Biggr\},
\end{equation*}
where $\varPi_{+}(S)=\bigcup_{a \geq0} \varPi_a(S)$.
\end{defin}

\begin{defin}\label{TSPg}
Let $g(S_t)=g^0+g^1S_t$, $g_0 \geq0$, $g_1 \geq0$. Define
\begin{equation*}
\varPi_g(S)=\bigl\{\pi\in \mathit{SF}(S): S_t^{\pi}
\geq-g(S_t), \ t \in[0,T]\bigr\}
\end{equation*}
and
\begin{equation*}
\varPsi_{g}(S)=\Biggl\{\psi\in L_g(\varOmega,
\mathcal{F}_T,\mathbb{P}) : \psi\leq \int_0^T(
\pi_s,dS_s) \ \text{for some} \ \pi\in
\varPi_g(S)\Biggr\},
\end{equation*}
where $L_g(\varOmega,\mathcal{F}_T,\mathbb{P})$ is the set of $\mathcal
{F}_T$-measurable random variables $\psi$ such that $|\psi|\leq g(S_T)$.
\end{defin}

We denote the closures of the sets $\varPsi_{+}(S)$ and $\varPsi_{g}(S)$ with
respect to norms $\|\cdot\|_{\infty}$ and $\|\cdot\|_g$ (see \cite
{book5}, p.~648, for definitions of these norms) by $\overline{\varPsi
}_{+}(S)$ and $\overline{\varPsi}_{g}(S)$, respectively.

Now following the notation presented in \cite{book5}, we proceed to the
main definitions of absence of arbitrage.

\begin{defin}\label{NAplus}
We say that the property $\overline{NA}_{+}$ (or equivalently that
the~market is $\overline{NA}_{+}$) holds if
\begin{equation*}
\overline{\varPsi}_{+}(S)\cap L^{+}_{\infty}(
\varOmega,\mathcal{F}_T,\mathbb {P})=\{0\},
\end{equation*}
where $L^{+}_{\infty}(\varOmega,\mathcal{F}_T,\mathbb{P})$ is the subset
of nonnegative random variables in
$L_{\infty}(\varOmega,\mathcal{F}_T,\mathbb{P})$.
\end{defin}

\begin{defin}\label{NAg}
We say that the property $\overline{NA}_g$ holds (or equivalently that
the~market is $\overline{NA}_g$) if
\begin{equation*}
\overline{\varPsi}_g(S)\cap L^{+}_{\infty}(
\varOmega,\mathcal{F}_T,\mathbb {P})=\{0\}.
\end{equation*}
\end{defin}

There are two theorems that establish necessary and sufficient
conditions for the absence of arbitrage in the market in terms of
equivalent (local) martingale measures. An important condition that
will be addressed further is the local boundedness of the price process.

\begin{defin}\label{EMMeas}
A probability measure $\mathbb{Q}$, which is equivalent to the
objective measure $\mathbb{P}$, is called an equivalent (local)
martingale measure if the discounted price process is a (local)
martingale under the measure $\mathbb{Q}$.
\end{defin}

\begin{defin}
A stochastic process $S$ is called locally bounded if there exists a
sequence $(\tau_n)_{n=1}^{\infty}$ of stopping times increasing a.s. to
$+\infty$ and such that the stopped processes $S_t^{\tau_n}=S_{t\wedge
\tau_n}$ are uniformly bounded for each $n \in\mathds{N}$.
\end{defin}

\begin{thm}\textup{(\cite{book5})}
\label{NAPlusCriteria}
Let a semimartingale $S$ be locally bounded. Then the market is
$\overline{NA}_{+}$ if and only if there exists an equivalent local
martingale measure \emph{(}ELMM\emph{)}.
\end{thm}

\begin{thm}\textup{(\cite{book5})}
\label{NAgCriteria}
Let a semimartingale $S$ be locally bounded. Then the market is
$\overline{NA}_g$ if and only if there exists an equivalent martingale
measure \emph{(}EMM\emph{)}.
\end{thm}

The following theorem is a corollary of Proposition~6.1 from \cite
{paper2} and defines the construction of ELMM in the model \eqref
{Model0}--\eqref{Model1}.

\begin{thm}\label{Frey}
A probability measure $\mathbb{Q}$, which is equivalent to the
objective measure $\mathbb{P}$ on $\mathcal{F}_T$, is an ELMM for the
process $S$ defined by the model~(\ref{Model0})--(\ref{Model1}) on
$\mathcal{F}_T$ if and only if there exists a progressively measurable
process $\nu=(\nu_t)_{0\leq t \leq T}$, $\int^T_0 \nu_s^2ds<\infty$
$\mathbb{P}$-a.s., such that the local martingale $(L_t)_{0\leq t \leq
T}$ defined by

\begin{align}
\label{Lt}
\notag L_t=d\mathbb{Q}/d\mathbb{P}|_{\mathcal{F}_t} &=\exp\Biggl( \int_0^t (r-\mu)/\sigma(Y_s)dB_s+\int_0^t \nu_sdZ_s\\
&\quad -\frac{1}{2}\int_0^t\bigl((r-\mu)^2/\sigma^2(Y_s)+\nu^2_s\bigr)ds \Biggr)
\end{align}
satisfies $\mathbb{E}L_T=1$.
\end{thm}

Denote by $\mathcal{LM}^S(\mathbb{P})$ and $\mathcal{M}^S(\mathbb{P})$
the sets of ELMM and EMM in the market modeled by \eqref{Model0}--\eqref
{Model1}. It is obvious that $\mathcal{M}^S(\mathbb{P}) \subset\mathcal
{LM}^S(\mathbb{P})$.

Recall that there is a decomposition of a $\mathbb{P}$-semimartingale
$S$ into the sum of a local $\mathbb{P}$-martingale $M$ and an adapted
finite-variation process $A$: $S=S_0+M+A$.

\begin{defin}
A probability measure $\mathbb{Q}$, which is equivalent to the
objective measure $\mathbb{P}$, is called a minimal martingale measure
(MMM) if $\mathbb{Q}=\mathbb{P}$ on~$\mathcal{F}_0$ and any
square-integrable $\mathbb{P}$-martingale strictly orthogonal to the
process $M$ is a local $\mathbb{Q}$-martingale.
\end{defin}

A minimal martingale measure is unique (see \cite{paper3}).

\section{Absence of arbitrage in the general model}\label{sec4}

In this section, we investigate the absence of arbitrage in the model
\eqref{Model0}--\eqref{Model1}. Notice that we further deal with an
undiscounted process $S$ defined in Section~\ref{sec2}.\vadjust{\eject}

\begin{thm}
The market defined by the model (\ref{Model0})--(\ref{Model1}) with
assumptions (A1)--(A3):

\begin{itemize}

\item[\rm(i)] satisfies $\overline{NA}_{+}$ property;

\item[\rm(ii)] satisfies $\overline{NA}_{g}$ property, provided that
for some ELMM $\mathbb{Q}$,
%
\begin{equation}
\label{Sigma2X2} \mathbb{E}^{\mathbb{Q}}\int_0^T
\sigma^2(Y_s)X^2_sds<\infty.
\end{equation}
\end{itemize}
\end{thm}
\begin{proof}
(i) Since $S$ is locally bounded due to its continuity, Theorem~\ref
{NAPlusCriteria} yields that in order to prove the first part of the
theorem, it suffices to show that $\mathcal{LM}^S(\mathbb{P}) \neq
\emptyset$.

Consider the process $L_t$ defined by \eqref{Lt} with $\nu=(\nu
_t)_{0\leq t \leq T}=0$. Let $L_T=d\mathbb{Q}/d\mathbb{P}|_{\mathcal
{F}_T}$. In view of Theorem~\ref{Frey}, it suffices to show that, under
such a choice of~$\nu$, we have
%
\begin{equation}
\label{ELT} \mathbb{E}L_T=1.
\end{equation}
It suffices to verify the Novikov condition
%
\begin{equation}
\label{Novikov} \mathbb{E}\exp \Biggl(\frac{1}{2} \int_0^T
\bigl((r-\mu)^2/\sigma^2(Y_s)+\nu
^2_s\bigr)ds \Biggr)< \infty.
\end{equation}

It follows from the boundedness away from zero of the function $\sigma$
(assumption (A2)) and our choice of $\nu$ that inequality \eqref
{Novikov} holds. Hence, $\mathbb{Q} \subset\mathcal{LM}^S(\mathbb
{P})$, which proves part (i) of the theorem.

(ii) Now let us show that the measure $\mathbb{Q}$ from (i) is an EMM.
Denote $\alpha(s):=(r-\mu)/\sigma(Y_s)$. Knowing that the measure
$\mathbb{Q}$ is equivalent to the measure $\mathbb{P}$ and is defined
by the Radon--Nikodym derivative $L_T=\frac{d\mathbb{Q}}{d\mathbb
{P}}|_T$, we may apply the Girsanov theorem to derive that the
processes $B_t^{\mathbb{Q}}:=B_t-\int_0^t \alpha(s)ds$, $Z_t^{\mathbb
{Q}}:=Z_t$, $0 \leq t \leq T$, are Wiener processes w.r.t.\  $\mathbb
{Q}$. The asset price process under the measure $\mathbb{Q}$ is a
solution of the stochastic differential equation
\begin{equation*}
dS_t=rS_tdt+\sigma(Y_t)S_tdB^{\mathbb{Q}}_t,
\end{equation*}
which yields the following representation for the discounted price
process $X_t:=\operatorname{e}^{-rt}S_t$:
%
\begin{equation}
\label{X_t} X_t=S+\int_0^t
\sigma(Y_s)X_sdB^{\mathbb{Q}}_s.
\end{equation}

Hence, provided that assumption (A2), as mentioned before, yields the
square integrability of $\sigma(Y_s)X_s$ on $[0,T]$, $X_t$ is a martingale.
Therefore, $\mathbb{Q} \subset\mathcal{M}^S(\mathbb{P})$, and by
Theorem~\ref{NAgCriteria} we deduce that the market $(S,Y)$ is
$\overline{NA}_g$.
\end{proof}

In order to prove $\overline{NA}_g$ and $\overline{NA}_{+}$ properties
of the market, we have proved the existence of one EMM. However, we
actually have a family of equivalent martingale measures in the market.
The discounted price process is of the form~\eqref{X_t} and thus is
always a martingale given the square integrability of $\sigma(Y_s)X_s$
and any \emph{admissible} choice of process $\nu$.

\begin{lemma}Let the market defined by (\ref{Model0})--(\ref{Model1})
with assumptions (A1)--(A3) and additional condition (\ref{Sigma2X2}), and let the measure $\mathbb{Q}$ be such that
\begin{equation*}
\mathbb{E}^{\mathbb{Q}}\exp \Biggl(\frac{1}{2} \int_0^T
\nu^2_sds \Biggr)< \infty.
\end{equation*}

Then $\mathbb{Q} \subset\mathcal{M}^S(\mathbb{P})$.
\end{lemma}

Since we have more than one equivalent martingale measure in the
market, it is straightforward that the market is incomplete. Each EMM
in the market is defined by the process $\nu(s)=\nu(s,Y_s,S_s)$
associated with it. In the financial literature, the process $\nu_s$ is
called the market price of volatility risk.

Under the EMM $\mathbb{Q}^{\nu}$, the pair of processes $(S_t, Y_t)$
has the following representation:\vspace*{-9pt}
%
%
\begin{align}
dS_t & =  r S_tdt+\sigma(Y_t)S_tdB^\mathbb{Q}_t, \notag\\
dY_t & =   \bigg( -\alpha Y_t-k  \bigg( \rho\dfrac{\mu-r}{\sigma(Y_t)}+\nu(t) \sqrt{1-\rho^2} \bigg)  \bigg) dt+kdW^\mathbb{Q}_t,\label{ModelQ}
\end{align}
%
where the processes
\begin{align*}
B^\mathbb{Q}_t & = B_t+\int_0^t\dfrac{\mu-r}{\sigma(Y_s)}ds,\\
W^\mathbb{Q}_t & = \rho B^\mathbb{Q}_t+\sqrt{1-\rho^2}Z^\mathbb {Q}_t,\quad\mbox{and}\\
Z^\mathbb{Q}_t & = Z_t+\int_0^t\nu(s)ds
\end{align*}
are Wiener processes w.r.t.\ $\mathbb{Q}$ according to the Girsanov
theorem, and $B^\mathbb{Q}$ and~$Z^\mathbb{Q}$ are independent.

In the risk-neutral model (\ref{ModelQ}), the volatility process is not
the Ornstein--Uhlenbeck process anymore. So, generally speaking, there
is no analytic solution to the corresponding differential equation.
Therefore, we further consider a particular case of the general model,
which is defined by a set of assumptions concerning the form and
behavior of certain parameters of the model.

\section{Case of uncorrelated processes}\label{sec5}

Let us define a modified set of assumptions:

(B1) The Wiener processes $B $ and $W $ are independent, that is, $\rho=0$;

(B2) $=$ (A2);

(B3) $=$ (A3).

Assumption (B1) simplifies the risk-neutral model to the following form:
%
%
\begin{align}
dS_t & =  r S_tdt+\sigma(Y_t)S_tdB^\mathbb{Q}_t, \notag\\
dY_t & = \big(-\alpha Y_t-k \nu(t) \big)dt+kdZ^\mathbb{Q}_t,\label{ModelB}
\end{align}
%
where
\begin{align*}
B^\mathbb{Q}_t & = B_t+\int_0^t
\dfrac{\mu-r}{\sigma(Y_s)}ds\quad\mbox {and}
\\
Z^\mathbb{Q}_t & = Z_t+\int_0^t
\nu(s)ds
\end{align*}
are independent Wiener processes w.r.t. $\mathbb{Q}$.

Our purpose is to price a European call option in the model (\ref{ModelB}).
We limit further investigation to the valuation w.r.t.\ the minimal
martingale measure.

\begin{thm}
The EMM $\mathbb{Q}$ in the market defined by the model (\ref{ModelB})
is minimal iff the process $\nu$ corresponding to $\mathbb{Q}$ is
identically zero.
\end{thm}
\begin{proof}
Suppose $\nu(t)=0$, $t \in[0, T]$. Let $B$ and $B^\mathbb{Q}$ be
$\mathcal{F}_t$-adapted Wiener processes w.r.t.\ measures $\mathbb{P}$
and $\mathbb{Q}$, respectively. If $N_t$ is a square-integrable $\mathbb
{P}$-martingale, then we can apply the Kunita--Watanabe decomposition
to derive
\begin{equation*}
N_t=N_0+\int_0^t
l_udB_u+\int_0^t
l^\mathbb{Q}_udB^\mathbb{Q}_u+Z_t,
\end{equation*}
where
\begin{equation*}
\langle B, Z \rangle=\bigl\langle B^\mathbb{Q}, Z \bigr\rangle= 0.
\end{equation*}

Let $N$ be strictly orthogonal to $\int_0^t \sigma(u)dB_u$. Then
\begin{equation*}
0=\Biggl\langle N, \int_0^{\cdot}
\sigma_udB_u \Biggr\rangle_t=\int
_0^tl_u\sigma_udu,
\end{equation*}
where $l_t=0$, $t \in[0, T]$, a.s.. Then for $L_t=d\mathbb{Q}/d\mathbb
{P}|_t$ and $\gamma_t:=(r-\mu)/\sigma_t$,
\begin{align*}
d(N_tL_t)& =  N_tdL_t+L_tdN_t+d\langle N, L \rangle_t \\
& =  N_tdL_t+L_tdN_t+\gamma_tl_tdt \\
& =  N_tdL_t+L_tdN_t.
\end{align*}

The process $N_tL_t$ is a local $\mathbb{P}$-martingale; hence, $N_t$
is a local $\mathbb{Q}$-martingale. By definition $\mathbb{Q}$ is MMM.

The converse statement of the theorem comes straightforward from the
uniqueness of MMM.
\end{proof}

The solution of the differential equation defining the evolution of the
price of asset has the following representation:
%
\begin{equation}
\label{S_t} S_t=S\exp\Biggl\{rt+\int_0^t
\sigma(Y_s)dB^{\mathbb{Q}}_s-\frac{1}{2}\int
_0^t \sigma^2(Y_s)ds
\Biggr\}, \quad 0\leq t \leq T.
\end{equation}

For a fixed trajectory of $Y_s$, the argument of the exponential
function in the right-hand side of \eqref{S_t} is a Gaussian process,
and $S_t$, $1\leq t\leq T$, has the log-normal distribution with $\ln
S_t\sim N (\ln S+(r-\frac{1}{2}\bar{\sigma}_t^2)t,\bar{\sigma
}_t^2t )$, where $\hat{\sigma}^2_t=\hat{\sigma}^2_t(Y_s):=\frac
{1}{t}\int_0^t\sigma^2(Y_s)ds$, $0 \leq t \leq T$. This fact is crucial
for the derivation of expression for the value of European call option.

The value of European call option at time 0 w.r.t. the MMM is defined
by the general formula
\begin{equation*}
V_0=\operatorname{e}^{-rT}E^{\mathbb{Q}}
\bigl(S^{\mathbb{Q}}_T-K\bigr)^{+}.
\end{equation*}

We apply the telescopic property of mathematical expectation to
transform the previous expression as follows:
%
\begin{equation}
\label{EV0} V_0=\operatorname{e}^{-rT}
\mathbb{E}^{\mathbb{Q}}\bigl\{\mathbb{E}^{\mathbb
{Q}}\bigl\{
\bigl(S^{\mathbb{Q}}_T-K\bigr)^{+}\big|Y_s, 0
\leq s\leq T\bigr\}\bigr\}.
\end{equation}

The inner expectation is conditional on the path of $Y_s,\, 0 \leq s
\leq T$, and therefore, it is actually the Black--Scholes price for a
model with deterministic time-dependent volatility. According to Lemma
2.1 in \cite{paper4}, the inner expectation in \eqref{EV0} has the
following representation:
\begin{align}
\notag &\mathbb{E}^{\mathbb{Q}}\bigl\{\bigl(S^{\mathbb{Q}}_T-K\bigr)^{+}|Y_s, 0 \leq s\leq T\bigr\}\\
\notag &\quad =\operatorname{e}^{\ln{S}+rT}\varPhi \biggl(\frac{\ln S+(r+\frac{1}{2}\bar{\sigma}^2_0)T-\ln K}{\bar{\sigma}_0\sqrt{T}}\biggr)\\
\label{IntExp} &\qquad -K\varPhi \biggl(\frac{\ln S+(r-\frac{1}{2}\bar{\sigma}^2_0)T-\ln K}{\bar{\sigma}_0\sqrt{T}} \biggr),
\end{align}
where $\bar{\sigma}_t:=\sqrt{\dfrac{1}{T}\int_t^T\sigma^2(Y_s)ds} \geq
0$, and $\varPhi$ is the standard normal distribution function. Notice
that $\bar{\sigma}^2_0(Y_s)=\hat{\sigma}^2_T(Y_s)$. The former notation
may be viewed as the volatility averaged from the current moment to
maturity, whereas the latter is the volatility averaged from the
initial moment to the current one.

Notice that the inner conditional expectation is an increasing function
of~$\bar{\sigma}^2_0$ (see Lemma~3.1 in \cite{paper4}), which is the
type of behavior one may expect to be exhibited by the Black--Scholes
price of European call option.

Taking into account the form of inner integral, in order to derive an
analytic expression for the price of an option $V_0$, it is necessary
to deal with expectation of $\varPhi$. Instead of trying to evaluate the
integral analytically, it is possible to use the Monte Carlo method.

\section{Derivation of analytic expression for the option price}\label{sec6}

From Eqs.\  \eqref{EV0}--\eqref{IntExp} we can see that in order to
derive the formula for the option price, it is necessary to present the
exact formula for the expectation of $\varPhi$. In this section, we apply
the inverse Fourier transform after rearranging of the right-hand side
of \eqref{IntExp}.

We introduce the following deterministic functions $\sigma_i=\sigma
_i(s)$, $i=\overline{1,4}$:
\begin{gather}
\label{sigma12} \sigma_{1,2}(s)=\frac{s}{\sqrt{T}}\mp\frac{\sqrt{s^2T-2T(\ln
{(S/K)}+rT)}}{T},
\\
\label{sigma34} \sigma_{3,4}(s)=\frac{-s}{\sqrt{T}}\mp\frac{\sqrt{s^2T+2T(\ln{(S/K)}+rT)}}{T}.
\end{gather}
We define the domains of each of these functions to guarantee the
nonnegativity of the expressions under square root, that is, $s^2T \geq
2T(\ln{(S/K)}+rT$ for $\sigma_1$, $\sigma_2$, and $s^2T \geq-2T(\ln
{(S/K)}+rT$ for $\sigma_3$, $\sigma_4$.

\begin{lemma}\label{EV0Int}
Suppose that the market is defined by the model (\ref{ModelB}) with
assumptions (B1)--(B3), $\mathbb{Q}$ is MMM, and $V_0$ is the
price of European call option at time 0. Then we have the
following representations:
\begin{itemize}
\item[\rm1)] for $\ln{(S/K)}+rT \geq0$ and $k =\sqrt{2(\ln{(S/K)}+rT)}$,
%
\begin{align}
\notag V_0&=S\operatorname{e}^{rT} \Biggl(\varPhi(k)+\frac{1}{\sqrt{2\pi}}\int_{k}^{\infty} \bigl(\mathbb{Q}\bigl(\bar{\sigma}_0<\sigma_1(s)\bigr)
 +\mathbb{Q}\bigl(\bar{\sigma}_0>\sigma_2(s)\bigr) \bigr)\operatorname {e}^{-s^2/2}ds \Biggr)\\
\notag &\quad -K \Biggl(\varPhi(0)+\frac{1}{\sqrt{2\pi}} \Biggl(\int_{0}^{\infty}\mathbb {Q}\bigl(\bar{\sigma}_0<\sigma_4(s)\bigr)\operatorname{e}^{-s^2/2}ds\\
&\quad -\int_{-\infty}^{0}\mathbb{Q}\bigl(\bar{\sigma}_0>\sigma_4(s)\bigr)\operatorname{e}^{-s^2/2}ds \Biggr) \Biggr);\label{EV0_PSimple1}
\end{align}
%
\item[\rm2)] for $\ln{(S/K)}+rT < 0$ and $l =\sqrt{-2(\ln{(S/K)}+rT)}$,
%
\begin{align}
\notag V_0&=S\operatorname{e}^{rT} \Biggl(\frac{1}{2}+\frac{1}{\sqrt{2\pi}}\Biggl(\int_{0}^{\infty}\mathbb{Q}\bigl(\bar{\sigma}_0>\sigma_2(s)\bigr)\operatorname{e}^{-s^2/2}ds\\
\notag &\quad -\int_{-\infty}^{0}\mathbb{Q}\bigl(\bar{\sigma}_0<\sigma_2(s)\bigr)\operatorname{e}^{-s^2/2}ds \Biggr) \Biggr)\\
\notag &\quad -K \Biggl(\varPhi(-l)-\frac{1}{\sqrt{2\pi}}\int_{-\infty}^{-l}\bigl(\mathbb {Q}\bigl(\bar{\sigma}_0<\sigma_3(s)\bigr)\\
&\quad +\mathbb{Q}\bigl(\bar{\sigma}_0>\sigma_4(s)\bigr)\bigr)\operatorname {e}^{-s^2/2}ds \Biggr).\label{EV0_PSimple2}
\end{align}
%
\end{itemize}
\end{lemma}
\begin{proof}
From \eqref{EV0} and \eqref{IntExp} we have:
\begin{equation*}
V_0=S\operatorname{e}^{rT}\mathbb{E}^{\mathbb{Q}}\bigl(
\varPhi(d_1)\bigr)-K\mathbb {E}^{\mathbb{Q}}\bigl(
\varPhi(d_2)\bigr),
\end{equation*}
where $d_1$ and $d_2$ are defined as follows:
%
\begin{equation}
\label{d1d2} d_1=\frac{\ln S+(r+\frac{1}{2}\bar{\sigma}^2_0)T-\ln K}{\bar{\sigma}_0\sqrt{T}}, \qquad
d_2=d_1-\bar{\sigma}_0\sqrt{T}.
\end{equation}
Since
%
%
\begin{align}
\notag\varPhi(d_1)&=\frac{1}{\sqrt{2\pi}}\int_{-\infty}^{d_1}\operatorname {e}^{-s^2/2}ds\\
&=\frac{1}{2}+\frac{1}{\sqrt{2\pi}}\operatorname{I}_{\{d_1>0\}}\int_{0}^{d_1}\operatorname{e}^{-s^2/2}ds -\frac{1}{\sqrt{2\pi}}\operatorname{I}_{\{d_1<0\}}\int_{d_1}^{0}\operatorname{e}^{-s^2/2}ds,
\end{align}
%
we have
\begin{equation*}
\begin{aligned}
\mathbb{E}^{\mathbb{Q}}\bigl(\varPhi(d_1)\bigr) &=\frac{1}{2}+\frac{1}{\sqrt{2\pi}}\int_{0}^{\infty}\mathbb {Q}(s<d_1)\operatorname{e}^{-s^2/2}ds\\
&\quad -\frac{1}{\sqrt{2\pi}}\int_{-\infty}^{0}\mathbb{Q}(s>d_1)\operatorname {e}^{-s^2/2}ds.
\end{aligned}
\end{equation*}

The probabilities in the integrands may be represented as follows:
\begin{equation*}
\begin{gathered} \mathbb{Q}(s<d_1)=\mathbb{Q} \biggl(
\frac{1}{2}\bar{\sigma}^2_0T-s\bar {
\sigma}_0\sqrt{T}+\ln{(S/K)}+rT>0 \biggr),
\\
\mathbb{Q}(s>d_1)=\mathbb{Q} \biggl(
\frac{1}{2}\bar{\sigma}^2_0T-s\bar {
\sigma}_0\sqrt{T}+\ln{(S/K)}+rT<0 \biggr).
\\
\end{gathered} %
\end{equation*}

Similarly, for $\varPhi(d_2)$, we have
\begin{equation*}
\begin{aligned}
\varPhi(d_2)&=\frac{1}{\sqrt{2\pi}}\int_{-\infty}^{d_2}\operatorname {e}^{-s^2/2}ds\\
&=\frac{1}{2}+\frac{1}{\sqrt{2\pi}}\operatorname{I}_{\{d_2>0\}}\int_{0}^{d_2}\operatorname{e}^{-s^2/2}ds -\frac{1}{\sqrt{2\pi}}\operatorname{I}_{\{d_2<0\}}\int_{d_2}^{0}\operatorname{e}^{-s^2/2}ds.
\end{aligned} %
\end{equation*}

Hence,
\begin{equation*}
\begin{aligned}
\mathbb{E}^{\mathbb{Q}}\bigl(\varPhi(d_2)\bigr) &=\frac{1}{2}+\frac{1}{\sqrt{2\pi}}\int_{0}^{\infty}\mathbb {Q}(s<d_2)\operatorname{e}^{-s^2/2}ds\\
&\quad -\frac{1}{\sqrt{2\pi}}\int_{-\infty}^{0}\mathbb{Q}(s>d_2)\operatorname {e}^{-s^2/2}ds.
\end{aligned}
\end{equation*}

The probabilities from the integrands may be represented as follows:
\begin{equation*}
\begin{gathered} \mathbb{Q}(s<d_2)=\mathbb{Q} \biggl(
\frac{1}{2}\bar{\sigma}^2_0T+s\bar {
\sigma}_0\sqrt{T}-\ln{(S/K)}-rT<0 \biggr),
\\
\mathbb{Q}(s>d_2)=\mathbb{Q} \biggl(
\frac{1}{2}\bar{\sigma}^2_0T+s\bar {
\sigma}_0\sqrt{T}-\ln{(S/K)}-rT>0 \biggr).
\\
\end{gathered} %
\end{equation*}

Solutions of the quadratic equations, which correspond to the above
quadratic inequalities, do not necessarily exist; therefore, we
consider different cases:
\begin{itemize}

\item[1)] The discriminant $D_{12}:=s^2T-2T(\ln{(S/K)}+rT)$ is a
quadratic form w.r.t.~$s$. There are two possibilities:

\begin{itemize}
\item[1.1)] $\ln{(S/K)}+rT>0$. Then for $k=\sqrt{2(\ln{(S/K)}+rT)}$:
$D_{12}<0$, $s\in(-k;k)$; $D_{12}>0,\ s \in(-\infty;-k) \cup(k;\infty
)$; so
\begin{equation*}
\begin{aligned}
\mathbb{E}^{\mathbb{Q}}\bigl(\varPhi(d_1)\bigr) &=\frac{1}{2}+\frac{1}{\sqrt{2\pi}} \Biggl(\int_{0}^{k}\operatorname {e}^{-s^2/2}ds+\int_{k}^{\infty}\bigl(\mathbb{Q}\bigl(\bar{\sigma}_0<\sigma _1(s)\bigr)\\
&\quad +\mathbb{Q}\bigl(\bar{\sigma}_0>\sigma_2(s)\bigr) \bigr)\operatorname {e}^{-s^2/2}ds \Biggr)\\
&\quad -\frac{1}{\sqrt{2\pi}}\int_{-\infty}^{-k} \bigl(\mathbb{Q}\bigl(\bar{\sigma }_0\,{<}\,\sigma_2(s)\bigr)\,{-}\,\mathbb{Q}\bigl(\bar{\sigma}_0\,{<}\,\sigma_1(s)\bigr) \bigr)\operatorname{e}^{-s^2/2}ds.
\end{aligned} %
\end{equation*}

\item[1.2)] $\ln{(S/K)}+rT \leq0$. Then for any $s \in(-\infty;
\infty), \ D_{12}>0$, So
\begin{equation*}
\begin{aligned}
\mathbb{E}^{\mathbb{Q}}\bigl(\varPhi(d_1)\bigr) &=\frac{1}{2}+\frac{1}{\sqrt{2\pi}}\int_{0}^{\infty}\bigl(\mathbb{Q}\bigl(\bar {\sigma}_0<\sigma_1(s)\bigr)\\
&\quad +\mathbb{Q}\bigl(\bar{\sigma}_0>\sigma_2(s)\bigr) \bigr)\operatorname {e}^{-s^2/2}ds\\
&\quad -\frac{1}{\sqrt{2\pi}}\int_{-\infty}^{0} \bigl(\mathbb{Q}\bigl(\bar{\sigma }_0\,{<}\,\sigma_2(s)\bigr)\,{-}\,\mathbb{Q}\bigl(\bar{\sigma}_0\,{<}\,\sigma_1(s)\bigr) \bigr)\operatorname{e}^{-s^2/2}ds.
\end{aligned} %
\end{equation*}
\end{itemize}

\item[2)] The discriminant $D_{34}:=s^2T+2T(\ln{(S/K)}+rT)$ is a
quadratic form w.r.t.~$s$. There are two possibilities:
\begin{itemize}

\item[2.1)] $\ln{(S/K)}+rT<0$. Then for $l=\sqrt{-2(\ln{(S/K)}+rT)}$,
$D_{34}<0$, $s\in(-l;l)$; $D_{34}>0,\ s \in(-\infty;-l) \cup(l;\infty
)$. So
\begin{equation*}
\begin{aligned}
\mathbb{E}^{\mathbb{Q}}\bigl(\varPhi(d_2)\bigr) &=\frac{1}{2}+\frac{1}{\sqrt{2\pi}}\int_{l}^{\infty}\bigl(\mathbb{Q}\bigl(\bar {\sigma}_0<\sigma_4(s)\bigr)\\
&\quad -\mathbb{Q}\bigl(\bar{\sigma}_0<\sigma_3(s)\bigr)\bigr)\operatorname{e}^{-s^2/2}ds -\frac{1}{\sqrt{2\pi}} \Biggl(\int_{-l}^{0}\operatorname{e}^{-s^2/2}ds\\
&\quad +\int_{-\infty}^{-l} \bigl(\mathbb{Q}\bigl(\bar{\sigma}_0<\sigma_3(s)\bigr)+\mathbb {Q}\bigl(\bar{\sigma}_0>\sigma_4(s)\bigr) \bigr)\operatorname{e}^{-s^2/2}ds\Biggr).
\end{aligned} %
\end{equation*}

\item[2.2)] $\ln{(S/K)}+rT \geq0$. Then for any $s \in(-\infty;
\infty), \ D_{34}>0$. So
\begin{equation*}
\begin{aligned}
\mathbb{E}^{\mathbb{Q}}\bigl(\varPhi(d_2)\bigr) &=\frac{1}{2}+\frac{1}{\sqrt{2\pi}}\int_{0}^{\infty}\bigl(\mathbb{Q}\bigl(\bar {\sigma}_0<\sigma_4(s)\bigr)\\
&\quad -\mathbb{Q}\bigl(\bar{\sigma}_0<\sigma_3(s)\bigr)\bigr)\operatorname {e}^{-s^2/2}ds\\
&\quad -\frac{1}{\sqrt{2\pi}}\int_{-\infty}^{0} \bigl(\mathbb{Q}\bigl(\bar{\sigma }_0\,{<}\,\sigma_3(s)\bigr)\,{+}\,\mathbb{Q}\bigl(\bar{\sigma}_0\,{>}\,\sigma_4(s)\bigr) \bigr)\operatorname{e}^{-s^2/2}ds.
\end{aligned} %
\end{equation*}
\end{itemize}
\end{itemize}

Combining these cases, we get the following expressions for the option price:
\begin{itemize}
\item[1)] for $\ln{(S/K)}+rT \geq0$,
%
\begin{align}
\notag V_0&=S\operatorname{e}^{rT} \Biggl(\varPhi(k)+\frac{1}{\sqrt{2\pi}} \Biggl(\int_{k}^{\infty} \bigl(\mathbb{Q}\bigl(\bar{\sigma}_0<\sigma_1(s)\bigr)\\
\notag &\quad +\mathbb{Q}\bigl(\bar{\sigma}_0>\sigma_2(s)\bigr) \bigr)\operatorname {e}^{-s^2/2}ds\\
\notag &\quad -\int_{-\infty}^{-k} \bigl(\mathbb{Q}\bigl(\bar{\sigma}_0<\sigma_2(s)\bigr)-\mathbb {Q}\bigl(\bar{\sigma}_0<\sigma_1(s)\bigr) \bigr)\operatorname{e}^{-s^2/2}ds\Biggr) \Biggr)\\
\notag &\quad -K \Biggl(\frac{1}{2}+\frac{1}{\sqrt{2\pi}} \Biggl(\int_{0}^{\infty}\bigl(\mathbb{Q}\bigl(\bar{\sigma}_0<\sigma_4(s)\bigr) -\mathbb{Q}\bigl(\bar{\sigma}_0<\sigma_3(s)\bigr) \bigr)\operatorname {e}^{-s^2/2}ds\\
&\quad -\int_{-\infty}^{0} \bigl(\mathbb{Q}\bigl(\bar{\sigma}_0<\sigma_3(s)\bigr)+\mathbb {Q}\bigl(\bar{\sigma}_0>\sigma_4(s)\bigr) \bigr)\operatorname{e}^{-s^2/2}ds\Biggr) \Biggr);\label{EV0_P1}
\end{align}

\item[2)] for $\ln{(S/K)}+rT < 0$,
%
\begin{align}
\notag V_0&=S\operatorname{e}^{rT} \Biggl(\frac{1}{2}+\frac{1}{\sqrt{2\pi}}\Biggl(\int_{0}^{\infty} \bigl(\mathbb{Q}\bigl(\bar{\sigma}_0<\sigma_1(s)\bigr)\\[-1.5pt]
\notag &\quad +\mathbb{Q}\bigl(\bar{\sigma}_0>\sigma_2(s)\bigr) \bigr)\operatorname {e}^{-s^2/2}ds\\[-1.5pt]
\notag &\quad -\int_{-\infty}^{0} \bigl(\mathbb{Q}\bigl(\bar{\sigma}_0<\sigma_2(s)\bigr)-\mathbb {Q}\bigl(\bar{\sigma}_0<\sigma_1(s)\bigr) \bigr)\operatorname{e}^{-s^2/2}ds\Biggr) \Biggr)\\[-1.5pt]
\notag &\quad -K \Biggl(\varPhi(-l)+\frac{1}{\sqrt{2\pi}} \Biggl(\int_{l}^{\infty}\bigl(\mathbb {Q}\bigl(\bar{\sigma}_0<\sigma_4(s)\bigr)\\[-1.5pt]
\notag &\quad -\mathbb{Q}\bigl(\bar{\sigma}_0<\sigma_3(s)\bigr)\bigr)\operatorname {e}^{-s^2/2}ds\\[-1.5pt]
&\quad -\int_{-\infty}^{-l} \bigl(\mathbb{Q}\bigl(\bar{\sigma}_0<\sigma_3(s)\bigr)+\mathbb {Q}\bigl(\bar{\sigma}_0>\sigma_4(s)\bigr) \bigr)\operatorname{e}^{-s^2/2}ds\Biggr) \Biggr).\label{EV0_P2}
\end{align}
%
\end{itemize}
%

Recalling that $\bar{\sigma}_0 \geq0$ and noticing that some of the
probabilities presented are identically zero, we simplify \eqref
{EV0_P1} and \eqref{EV0_P2} to the forms \eqref{EV0_PSimple1} and \eqref
{EV0_PSimple2}, respectively.\looseness=1
\end{proof}

Let $S_i \subset\mathds{R}$ be the domains of positivity of functions
$\sigma_i(s)$, $i=\overline{1,4}$. It is easy to check that the
functions appearing in the integrals \eqref{EV0_PSimple1}--\eqref
{EV0_PSimple2} are positive on the integration domains.

Assume that the probability density function of $\bar{\sigma}^2_0$ is
piecewise continuous on $\mathds{R}$. Then due to the Fourier inversion
theorem, for almost all $s \in S_i$, the probabilities in the
integrands in \eqref{EV0_PSimple1}--\eqref{EV0_PSimple2} have the
following representation:
%
%
\begin{align}
\notag\mathbb{Q}\bigl(\bar{\sigma}_0<\sigma_i(s)\bigr)&=\mathbb{Q}\bigl(\bar{\sigma }^2_0<\sigma^2_i(s)\bigr)\\
&=\lim_{\substack{\varepsilon\to0}} \frac{1}{2 \pi} \int_{-\infty}^{\sigma^2_i(s)}\Biggl(\int_{-\infty}^{\infty} \operatorname{exp} \biggl(iyu-\frac{\varepsilon^2 u^2}{2} \biggr) \phi(u)du \Biggr)dy,
\end{align}
%
where $\phi(u)=\mathbb{E}^{\mathbb{Q}}(\operatorname{e}^{iu \bar{\sigma
}^2_0})$ is the characteristic function of $\bar{\sigma}^2_0$.

We are now in a position to state the main result of this section.

\begin{thm}\label{OptionPriceAnRep}
Suppose that the market is defined by the model (\ref{ModelB}) with
assumptions (B1)--(B3), $\mathbb{Q}$ is the MMM, and $V_0$ is
the price at time 0 of European call option. Let the probability
density function of $\bar{\sigma}^2_0$ be piecewise continuous on
$\mathds{R}$. Then we have the following representations:
\begin{itemize}
\item[1)] for $\ln{(S/K)}+rT \geq0$ and $k =\sqrt{2(\ln{(S/K)}+rT)}$,
%
%
\begin{align*}
V_0=&\lim_{\substack{\varepsilon\to0}} \Biggl(S\operatorname{e}^{rT} \Biggl(\varPhi(k)+\frac{1}{(2\pi)^{3/2}}  \\
&\times\Biggl(\int_{k}^{\infty}\Biggl(\int_{-\infty}^{\sigma^2_1(s)}\int_{-\infty}^{\infty} \operatorname{exp} \biggl(iyu-\frac{\varepsilon^2 u^2}{2} \biggr) \phi(u)dudy\\
&+\int_{\sigma^2_2(s)}^{\infty}\int_{-\infty}^{\infty} \operatorname {exp} \biggl(iyu-\frac{\varepsilon^2 u^2}{2} \biggr) \phi(u)dudy \Biggr)\operatorname{e}^{-s^2/2}ds\Biggr) \Biggr)\\
&-K \Biggl(\frac{1}{2}+\frac{1}{(2\pi)^{3/2}} \\
&\times \Biggl(\int_{0}^{\infty}\int_{-\infty}^{\sigma^2_4(s)}\int_{-\infty}^{\infty}\operatorname {exp}\biggl(iyu-\frac{\varepsilon^2 u^2}{2} \biggr) \phi(u)dudy\operatorname {e}^{-s^2/2}ds\\
&+\int_{-\infty}^{0}\int_{\sigma^2_4(s)}^{\infty}\int_{-\infty}^{\infty} \operatorname{exp} \biggl(iyu-\frac{\varepsilon^2 u^2}{2} \biggr) \phi (u)dudy\operatorname{e}^{-s^2/2}ds \Biggr)\Biggr) \Biggr);
\end{align*}
%
\item[2)] for $\ln{(S/K)}+rT < 0$ and $l
=\sqrt{-2(\ln{(S/K)}+rT)}$,
%
%
\begin{align*}
V_0=&\lim_{\substack{\varepsilon\to0}} \Biggl(S\operatorname{e}^{rT} \Biggl(\frac{1}{2}+\frac{1}{(2\pi)^{3/2}}\\
&\times\Biggl(\int_{0}^{\infty}\int_{\sigma^2_2(s)}^{\infty}\int_{-\infty}^{\infty} \operatorname{exp} \biggl(iyu-\frac{\varepsilon^2 u^2}{2} \biggr) \phi(u)dudy\operatorname {e}^{-s^2/2}ds\\
&-\int_{-\infty}^{0}\int_{-\infty}^{\sigma^2_2(s)}\int_{-\infty}^{\infty} \operatorname{exp} \biggl(iyu-\frac{\varepsilon^2 u^2}{2} \biggr) \phi (u)dudy \operatorname{e}^{-s^2/2}ds \Biggr)\Biggr)\\
&-K \Biggl(\varPhi(-l)\,{-}\,\frac{1}{(2\pi)^{3/2}}\int_{-\infty}^{-l}\Biggl(\int_{-\infty}^{\sigma^2_3(s)} \int_{-\infty}^{\infty}\operatorname {exp} \biggl(iyu\,{-}\,\frac{\varepsilon^2 u^2}{2} \biggr) \phi(u)dudy\\
&+\int_{\sigma^2_4(s)}^{\infty} \int_{-\infty}^{\infty}\operatorname {exp} \biggl(iyu-\frac{\varepsilon^2 u^2}{2} \biggr) \phi(u)dudy \Biggr)\operatorname{e}^{-s^2/2}ds \Biggr) \Biggr),
\end{align*}
%
\end{itemize}
where
$\phi(u)=\mathbb{E}^{\mathbb{Q}}(\operatorname{e}^{iu \bar{\sigma
}^2_0})$ is the characteristic function of the random variable $\bar
{\sigma}^2_0, $ and
$\sigma_i=\sigma_i(s)$, $i=\overline{1,4}$, are of the form (\ref{sigma12})--(\ref{sigma34}).
\end{thm}

\begin{remark}
If $\bar{\sigma}_0 \in L^2(\mathds{R})$, then the limit in Theorem~\ref
{OptionPriceAnRep} may be moved inside the integrals. Thus, $\epsilon$
may be equated to zero, and the expression for the option price is simplified.
\end{remark}

\begin{remark}
Under the assumption that $\sigma$ is bounded, we can rewrite the
analytical expression in terms of moments of $\bar{\sigma}^2_0$.
\end{remark}

Indeed, in this case, $\bar{\sigma}^2_0$ is bounded as well, so the
characteristic function $\phi(u)$ admits the Taylor series expansion
around zero:
%
\begin{equation}
\phi(u)=1+\sum_{j=1}^{\infty}
\frac{i^j u^j}{j!}m_j,
\end{equation}
where $m_n$ is the $n$th moment of the random variable $\bar{\sigma
}^2_0$, and $i=\sqrt{-1}$.

The moments of the random variable $\bar{\sigma}^2_0$ can be
represented by applying the fact that the finite-dimensional
distributions of the Ornstein--Uhlenbeck process are Gaussian vectors.
Bearing in mind that the covariance matrix of the process $Y_s$ is
nondegenerate and consists of the elements of the form
%
\begin{align}
\bigl(\varSigma^{i,l}\bigr)_{i,l=1}^j&=\frac{k^2}{2\alpha} \exp{ \bigl({-\alpha(t_{i}+t_{l})} \bigr)}\bigl(\exp\bigl(2\alpha\min(t_{i},t_{l})\bigr)-1\bigr),\label{CovMatr}
\end{align}
%
we get the following representation for the moments of the random
variable~$\bar{\sigma}^2_0$:
%
%
\begin{align}
m_j&=\frac{1}{T^j}\int^{T}_{0}\ldots\int^{T}_{0}\int_{\mathds{R}^j}\frac{\sigma^2(y_{1})\ldots\sigma^2(y_{j})}{(2\pi)^{j/2} \vert\varSigma\vert^{1/2}}\operatorname{e}^{-\frac{1}{2}(\mathbf{y} - \mathbf{\mu})^{\top}\varSigma^{-1} (\mathbf{y} - \mathbf{\mu})}\mathbf{dy} dt_{1}\ldots dt_{j},
\end{align}
%
where
\[\mathbf{y}=(y_1,\dots,y_j),\qquad \mathbf{dy}=dy_1\times\cdots\times
dy_j,\qquad \mathbf{\mu}= (Y_0\operatorname{e}^{-\alpha y_1}, \dots,
Y_0\operatorname{e}^{-\alpha y_j} ).\]

We have demonstrated that there is an analytic solution to the problem
of pricing of European call option in the model. However, the resulting
formula is complicated and cumbersome.
Therefore, our further investigation will be aimed at comparison of
numeric results produced by it with approximate calculations and
possible simplifications.

\bibliographystyle{vmsta-mathphys}
%

\end{document}